\documentclass[doublecolumn]{IEEEtran}

\usepackage{stfloats}
\usepackage{amsfonts}
\usepackage{amssymb}
\usepackage{amsthm}
\usepackage{cite}
\usepackage{float}
\usepackage{color}
\usepackage{stfloats,fancyhdr}
\usepackage{amsmath,bm}
\usepackage{algorithm}
\usepackage{algorithmic}
\usepackage{multirow}
\usepackage{changepage}
\usepackage[normalem]{ulem}
\usepackage{amsthm}
\usepackage{balance}

\usepackage{bbm}

\newtheorem{theorem}{Theorem}
\newtheorem{lemma}{Lemma}

\newtheorem{corollary}{Corollary}
\newtheorem{definition}{Definition}
\newtheorem{assumption}{Assumption}
\newtheorem{remark}{Remark}
\newtheorem{example}{Example}

\usepackage{mathtools}
\mathtoolsset{showonlyrefs=true}

\ifCLASSINFOpdf
\usepackage[pdftex]{graphicx}
\DeclareGraphicsExtensions{.pdf,.jpeg,.png}
\else
\usepackage[dvips]{graphicx}
\DeclareGraphicsExtensions{.eps}
\fi

\usepackage{subfigure}
\usepackage{fancybox,dashbox}
\usepackage{authblk}
\usepackage{tabularx}

\usepackage{comment}

\newcommand{\myexpect}[1]{\mathsf{E}\left[#1\right]}

\newcommand{\myprob}[1]{\mathsf{Prob}\left[#1\right]}

\newcommand{\mycomment}[1]{{\color{black}{#1}}}

\newcommand\aug{\fboxsep=-\fboxrule\!\!\!\fbox{\strut}\!\!\!}

\usepackage[margin=0.51in]{geometry}
\begin{document}
		
\title{\huge Stability Conditions for Remote State Estimation of Multiple Systems over Multiple Markov Fading Channels}

\author{
	\vspace{-0.0cm}
	Wanchun Liu,~\IEEEmembership{Member,~IEEE}, Daniel E.\ Quevedo,~\IEEEmembership{Fellow,~IEEE}, 
	Karl Henrik Johansson,~\IEEEmembership{Fellow,~IEEE}, \par 
	Branka Vucetic,~\IEEEmembership{Fellow,~IEEE}, Yonghui Li,~\IEEEmembership{Fellow,~IEEE}
\vspace{-0.5cm}
}
\maketitle

\begin{abstract}
\let\thefootnote\relax\footnote{W. Liu, B. Vucetic, and Y. Li are with School of Electrical and Information Engineering, The University of Sydney, Australia.
	Emails:	\{wanchun.liu,\ yonghui.li,\ branka.vucetic\}@sydney.edu.au. 
D. E. Quevedo is with the School of Electrical Engineering and Robotics, Queensland University of Technology (QUT), Brisbane, Australia.	Email: daniel.quevedo@qut.edu.au.
K. H. Johansson is with School of Electrical Engineering and Computer Science, KTH Royal
Institute of Technology, Stockholm, Sweden. Email: kallej@kth.se.
}
We investigate the stability conditions for remote
state estimation of multiple linear time-invariant (LTI) systems over multiple wireless time-varying communication channels.
We answer the following open problem: what is the fundamental requirement on the multi-sensor-multi-channel system to guarantee the existence of a sensor scheduling policy that can stabilize the remote estimation system? We propose a novel policy construction and analytical framework and derive the necessary-and-sufficient stability condition in terms of the LTI system parameters and the channel statistics.
\end{abstract}


\vspace{-0.2cm}
\section{Introduction}
\subsection{Motivation}
\mycomment{Industry 4.0, also known as the Fourth Industrial Revolution, is the automation of traditional manufacturing and
industrial processes through customized and flexible mass production~\cite{Indusrtial40}. 
Replacing communication cables with wireless devices in conventional factories will be a game-changer: In particular, for automatic control, Industry
4.0 will make use of large-scale, interconnected deployment of massive spatially distributed industrial devices such as
sensors, actuators, and controllers. 
Given their  low-cost and scalable deployment, 
wireless remote state estimation from ubiquitous sensors
will play a key role in many industrial control applications, such as smart manufacturing, industrial automation,
e-commerce warehouses, and smart grids~\cite{ParkSurvey}.}

However, unlike wired communications, wireless communications 
are often unreliable and have a limited spectrum for transmission~\cite{tse2005fundamentals}.
Consequently, when wireless sensors are deployed in a remote estimation system,  scheduling policies need to be designed to allow sensors to update the measurement data over a limited number of frequency channels.
The design of such transmission schedules is especially challenging since, due to variability of the environment, wireless channels are time-varying. The transition process of associated channel fading states are commonly modeled as a Markov processes~\cite{Parastoo,markovchannel1,markovchannel2}, wherein different channel states lead to different packet drop probabilities.
Due to transmission scheduling and packet dropouts, a multi-sensor remote estimator cannot correctly receive all the sensor measurements. This degrades the estimation performance and can even lead to instability, i.e., the expected estimation error covariances may become unbounded.

In this note, we tackle the \textbf{fundamental problem: what are necessary and sufficient conditions on system parameters that ensure stochastic stability of a multi-sensor remote estimator under multiple Markov fading channels?}

\subsection{Related Works}
Existing work related to multi-system remote estimation and control can be divided into two categories: perfect (wired) and imperfect (wireless) communication channels.

\textbf{Perfect communication channels.}
Early research in stability analysis of multi-control-loop transmission scheduling over single and multiple
independent communication channels involved periodic and aperiodic scheduling policies~\cite{Hristu,Rehbinder,Zhang,Walsh,IWAKI}, assuming perfect communication channels.
These works have only determined the sufficient conditions to guarantee the existence of a scheduling policy that can stabilize the networked systems.

\textbf{Imperfect communication channels.}
In practice, wireless channels are not error-free, leading to transmission errors and packet dropouts.  In industrial control environments   comprising moving machines and mobile robots, 
the channel quality is time-varying~\cite{papermill}.
Unlike single systems for which remote estimation and control that have been well investigated (e.g., \cite{Schenato07ProcIEEE,schenato2008optimal} for independent and identically distributed packet dropout scenarios and~\cite{liu2020remote} for  Markovian packet dropout), 
multi-system remote estimation and control over wireless channels have not drawn as much research attention, until recent efforts motivated by  standardization and growing deployment of wireless technology.

Considering a wireless control architecture with multiple control loops over shared wireless fading channels,   optimal dynamic transmission scheduling policies were investigated in~\cite{gatsis2015opportunistic,Eisen} with different objective functions. The scheduling decision in each step depends on both the wireless channel and control loop states.
\mycomment{In~\cite{alexscheduling} and~\cite{Wu2018Auto}, sensor transmission scheduling over  single and multiple packet drop channels of remote estimators were investigated, respectively, and some sufficient mean-square stability conditions in terms of the system parameters and optimal transmission scheduling policies were analyzed.
If the stability conditions hold, classic Markov decision process (MDP) methods were adopted for finding the optimal scheduling policies (see e.g., \cite{KangTWC,KangJIoT,Huang2021Length}).}
The follow-up work~\cite{LEONG2020108759} considered a time-correlated Markov fading channel scenario, and derived a sufficient condition to guarantee the existence of a deterministic and stationary scheduling policy that can stabilize the remote estimator. A deep reinforcement learning method was proposed as well to find the optimal scheduling policy. The approach was further applied for solving a transmission scheduling problem of a fully distributed networked control system in~\cite{liu2021DRL}.
In \cite{EGWPeters}, scheduling policy design of a spatially distributed large control system with many sensors and actuators based on the   wireless network standard IEEE 802.15.4 was investigated.

\mycomment{We also note that the conventional Markov jump linear systems theory~\cite{costa2006discrete} can be used to elucidate stability conditions of networked systems using a given transmission policy. However, MJLS theory does not provide much insights about the existence of a dynamic scheduling policy that can stabilize the networked system over shared communication channels.}

\subsection{Contributions}
{\color{black}In this note, we consider a multi-sensor remote estimator with multiple frequency channels, where individual sensors measure different physical processes.
We allow  the channels to be time-varying and correlated in both frequency and time, as is common in practical applications~\cite{tse2005fundamentals}. We derive a necessary and sufficient mean-square stability condition in terms of the physical process parameters and the fading channel statistics. Our result establishes that there exists at least one sensor scheduling policy over the frequency channels that can stabilize the remote estimator if and only if the condition holds. 
\mycomment{The stability condition will provide practical design guidelines for stabilizing multi-sensor remote estimation systems over shared wireless medium in Industry 4.0.}
 
It is worth emphasizing that the stability condition depends on the essential parameters of the physical processes of interest and the communication channels, rather than on the specific scheduling policy employed. 
In fact, to find merely sufficient conditions for stability, one can construct a specific policy and analyze its properties (e.g.,~\cite{LEONG2020108759,wu2018optimal}). Such a sufficient condition, however, is commonly not tight  and thus cannot be proved to be necessary.
The analysis of necessary conditions is challenging as there exists a combinational number of scheduling policies.

To the best of our knowledge, necessary and sufficient stability conditions have not been established before for remote estimation or control over multiple systems and wireless communication channels.
In the present work,  we consider a  general correlated channel, wherein the channel state in each frequency channel is a random variable and correlated with the other channel states.
To tackle the challenge, we develop a novel policy construction method and a stochastic estimation-cycle based analytical approach. 
We develop an asymptotic theory for the spectral radius of a product of non-negative matrices to prove our key result.}
%

\emph{Notations:} 
Sets are denoted by calligraphic capital letters, e.g., $\mathcal{A}$.
$\mathcal{A} \backslash \mathcal{B}$ denotes set subtraction.
Matrices and vectors are denoted by capital and lowercase upright bold letters, e.g., $\mathbf{A}$ and $\mathbf{a}$, respectively.
$\vert \mathcal{A}\vert$ denotes the cardinality of the set $\mathcal{A}$.
$\mathsf{E}\left[A\right]$ is the expectation of the random variable $A$.
$(\cdot)^{\top}$ is the matrix transpose operator. $\| \mathbf{v} \|_1$ is the sum of the vector $\mathbf{v}$'s elements. 
$|\mathbf{v}| \triangleq \sqrt{\mathbf{v}^\top \mathbf{v}}$ is the Euclidean norm of a vector $\mathbf{v}$.
$\text{Tr}(\cdot)$ is the trace operator. $\text{diag}\{\mathbf{v}\}$ denotes the diagonal matrix with the diagonal elements taken from $\mathbf{v}$. $\mathbb{N}$ and $\mathbb{N}_0$ denote the sets of positive and non-negative integers, respectively.
$\mathbb{R}^m$ denotes the $m$-dimensional Euclidean space.
$\left[u\right]_{B \times B}$ denotes the $B \times B$ matrix with identical elements~$u$.
$[\mathbf{A}]_{j,k}$ denotes the element at the $j$th row and $k$th column of a matrix $\mathbf{A}$.
$\{v\}_{\mathbb{N}_0}$ denotes the semi-infinite sequence   $\{v_0,v_1,\dots\}$.
$\rho(\mathbf{A})$ is the spectral radius of $\mathbf{A}$, i.e., the  largest absolute value of its eigenvalues.
$\vartheta(\mathbf{A})$ is the sum of all the elements of $\mathbf{A}$.

\section{System Model}\label{sec:sys}
We consider a remote estimator with $N$ sensors measuring $N$ independent physical process,   as illustrated in Fig.~\ref{fig:system_model}. The index set of the sensors is denoted as $\mathcal{N} \triangleq \{1,2,\dots,N\}$.
The sensors are connected with a local gateway, which collects the  measurements and forwards them to a remote estimator. 
Connections between sensors and the gateway are reliable and not scheduled, while the gateway to remote estimator communications are wireless and scheduled due to  bandwidth limitations.
\mycomment{We note that the typical connection density in the Industrial IoT scenario is $10^6$/km$^2$. Bandwidth sharing among a large number of wireless devices is an issue in practice~\cite{MassivIoT}.}
\begin{figure}[t]
	\centering\includegraphics[scale=0.45]{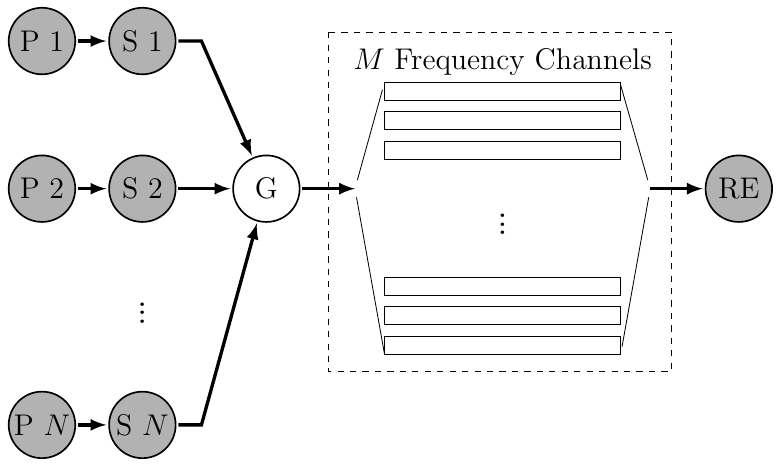}
	\vspace{-0.2cm}	
	\caption{The multi-sensor-multi-channel remote estimator with a single gateway. Processes, sensors, gateway and remote estimator are denoted as P$1$, P$2$, ..., P$N$,  S$1$, S$2$, ..., S$N$, G and RE, respectively.}
	\vspace{-0.2cm}
	\label{fig:system_model}
\end{figure}

The discrete-time linear time-invariant (LTI) model of the measurement of each process $n$ is given as~\cite{schenato2008optimal,shi2012optimal,yang2013schedule}
\begin{equation} \label{sys}
\begin{aligned}
\mathbf{x}_n{(t+1)} &= \mathbf{A}_n \mathbf{x}_n(t) + \mathbf{w}_n(t),\\
\mathbf{y}_n(t) &= \mathbf{C}_n\mathbf{x}_n(t) + \mycomment{\mathbf{z}_n(t)},
\end{aligned}
\end{equation}
where 
$\mathbf{x}_n(t) \in \mathbb{R}^{l_n}$ is the process state vector, $\mathbf{A}_n \in \mathbb{R}^{{l_n} \times {l_n}}$ the state transition matrix, $\mathbf{y}_n(t) \in \mathbb{R}^{r_n}$ the measurement vector of the sensor attached to the process, $\mathbf{C}_n \in \mathbb{R}^{{r_n} \times {l_n}}$ the measurement matrix, $\mathbf{w}_n(t) \in \mathbb{R}^{l_n}$ and $\mathbf{z}_n(t) \in \mathbb{R}^{r_n}$ are the process and measurement noise vectors, respectively. 
We assume $\mathbf{w}_n(t)$ and $\mathbf{z}_n(t)$ are independent and identically distributed (i.i.d.) zero-mean Gaussian processes with corresponding covariance matrices $\mathbf{W}_n$ and $\mathbf{Z}_n$, respectively.
Without loss of generality, we assume that 
$ \rho^2(\mathbf{A}_1) \geq \rho^2(\mathbf{A}_2)\geq \dots \geq \rho^2(\mathbf{A}_N)$. 

\subsection{Local Estimation}
Each sensor adopts a local Kalman filter (KF) to estimate its process before sending to the gateway~\cite{shi2012optimal,yang2013schedule,liu2020remote}.
We have
\begin{subequations}\label{sub:1}
	\begin{align}
	\mathbf{x}^s_n ({t|t-1})&=\mathbf{A}_n \mathbf{x}^s_n ({t-1})\\
	\mathbf{P}^s_n ({t|t-1})&=\mathbf{A}_n \mathbf{P}^s_n ({t-1}) \mathbf{A}^{\top}_n+\mathbf{W}_n\\
	\mathbf{K}_n (t)&=\mathbf{P}^s_n ({t|t-1}) \mathbf{C}^{\top}_n(\mathbf{C}_n \mathbf{P}^s_n ({t|t-1}) \mathbf{C}^{\top}_n+\mathbf{Z}_n)^{-1}\\
	\mathbf{x}^s_n (t)&=\mathbf{x}^s_n ({t|t-1})+\mathbf{K}_n (t)(\mathbf{y}_n({t})-\mathbf{C}_n \mathbf{x}^s_n ({t|t-1}))\\
	\mathbf{P}^s_n (t)&=(\mathbf{I}_n-\mathbf{K}_n ({t}) \mathbf{C}_n)\mathbf{P}^s_n ({t|t-1})
	\end{align}
\end{subequations}
where $\mathbf{I}_n$ is the $l_n \times l_n$ identity matrix, $\mathbf{x}^s_n ({t|t-1})$ is the prior state estimate, $\mathbf{x}^s_n (t)$ is the posterior state estimate at time $t$, $\mathbf{K}_n (t)$ is the Kalman gain. The matrices $\mathbf{P}^s_n ({t|t-1})$ and $\mathbf{P}^s_n (t)$ represent the prior and posterior error covariance at the sensor at time $t$, respectively. The first two equations above present the prediction steps while the last three equations correspond to the updating steps.
In particular, $\mathbf{x}^s_n (t)$ is the sensor $n$'s estimate of $\mathbf{x}_n(t)$ at time $t$, i.e., the pre-filtered measurement of~$\mathbf{y}_t$, with the estimation error covariance $\mathbf{P}^s_n (t)$ defined as:
$$ \mathbf{P}^s_n(t) \triangleq \myexpect{(\mathbf{x}^s_n(t)-\mathbf{x}_n(t))(\mathbf{x}^s_n(t)-\mathbf{x}_n(t))^\top}.$$
We focus on the \emph{remote} estimation stability and  assume that the \emph{local} KFs are stable and operate in   steady state~\cite{shi2012optimal,yang2013schedule,liu2020remote}, i.e., $\mathbf{P}^s_n(t) = \mathbf{\bar{P}}_n, \forall t\in \mathbb{N}, n \in \mathcal{N}$.

\subsection{Markov Channel}
We assume that there exist only $M<N$ frequency channels that can be used for transmission of sensor data.  The channels are correlated in both time and frequency domains as detailed below.

	The $M$-channel (vector) state $\mathbf{h}(t)$ is modeled as an aperiodic Markov chain with $\bar{M}$ irreducible channel states, $\mathcal{S} \triangleq \{\mathbf{h}_1,\mathbf{h}_2,\dots,\mathbf{h}_{\bar{M}}\}$, where $\mathbf{h}_i \triangleq [h_{i,1},h_{i,2},\dots,h_{i,M}]^\top \in \{0,1\}^M$.
Here, $h_{i,j}=0$ or $1$ means successful (on) or failed (off) transmission in the $j$th frequency channel at the $i$th channel state. 
Let $\mathbf{M} \in \mathbb{R}^{\bar{M}} \times \mathbb{R}^{\bar{M}}$ denote the state transition probability matrix, where
\begin{equation}
[\mathbf{M}]_{i,j}\triangleq \myprob{\mathbf{h}(t) =\mathbf{h}_j \vert \mathbf{h}(t-1) =\mathbf{h}_i} = p_{i,j}.
\end{equation}

Let $\tilde{\mathcal{S}_{i}} \in \mathcal{S}$ denote the set of channel state with an `on' state in the $i$th frequency channel, i.e.,
\begin{equation}
\tilde{\mathcal{S}_{i}} \triangleq \{\mathbf{h}_j: h_{j,i} =1, j \in \mathcal{\bar{M}}\}, i\in \mathcal{M},
\end{equation}
where $\mathcal{M} \triangleq \{1,2,\dots,M\}$ and $\mathcal{\bar{M}} \triangleq \{1,2,\dots,\bar{M}\}$.

We make the following assumption on the availability of the channel state.
\begin{assumption}[Known Previous Channel State~\cite{LEONG2020108759}]\label{ass:pre-channel}
	At time $t\in \mathbb{N}$, the current channel state is unknown but the previous channel state $\mathbf{h}(t-1)$ is available.
\end{assumption}

\subsection{Transmission Scheduling and Remote Estimation}
\sloppy In each time slot, the gateway collects $N$ packets carrying the sensor estimates $\{\mathbf{\hat{x}}^s_1(t),\dots,\mathbf{\hat{x}}^s_N(t)\}$.
It schedules at most $M$ of the packets and sends them through $M$ frequency channels to the remote estimator, as illustrated in Fig.~\ref{fig:system_model}. 
\mycomment{In practice, one could adopt a multiplexing scheme such as the orthogonal frequency-division multiplexing (OFDM) for transmitting multi-stream data in parallel.}
Each frequency channel can transmit at most one packet. The unscheduled packets are discarded.
\mycomment{The communication protocol for gateway-to-remote-estimator transmission is user datagram protocol (UDP)~\cite{Schenato07ProcIEEE}, which is widely adopted for real-time communications.}
We make the following assumption on transmission redundancy.
\begin{assumption}[Disabled Redundant Transmissions] \label{ass:no-redundent}
In each time slot, each packet can take at most one frequency channel for transmission.
\end{assumption}

Due to the transmission scheduling and the fading channels, packets carrying the estimated $N$ process states may or may not arrive at the remote estimator. Let $\gamma_n(t)=1$ denote successful detection of sensor $n$'s packet at time $t$, $n \in \mathcal{N}$.
If $\gamma_n(t)=0$, the packet is not scheduled or is scheduled but with failed detection.
It is also assumed that each packet transmission has a unit delay that is equal to the sampling period of the system.
The optimal remote estimator in the sense of minimum mean-square error (MMSE) is obtained as~\cite{schenato2008optimal}
\begin{equation}
\hat{\mathbf{x}}_n(t) = \begin{cases}
\mathbf{A}_n \hat{\mathbf{x}}_n(t-1),& \gamma_n(t-1)=0,\\
\mathbf{A}_n \hat{\mathbf{x}}_n^s(t-1),& \gamma_n(t-1)=1,
\end{cases}
\end{equation}
and can be simplified as~\cite{liu2020remote}
\begin{equation}\label{general_estimater}
\hat{\mathbf{x}}_n(t) = \mathbf{A}^{\phi_n(t)}_n \hat{\mathbf{x}}^s_n(t-(\phi_n(t))),
\end{equation}
where $\phi_n(t)\in \mathbb{N}$ is the time duration between the previous successful transmission and the current time $t$, and can be regarded  as the \emph{age-of-information} (AoI)~\cite{kaul2012real}.

From the above it follows that the estimation error covariance of process $n$ is given as
\begin{align} \label{covariance1}
\mathbf{P}_n(t) 
&\triangleq  \mathsf{E}\left[(\hat{\mathbf{x}}_n(t)-\mathbf{x}_n(t))(\hat{\mathbf{x}}_n(t)-\mathbf{x}_n(t))^{\top}\right]\\
&=\zeta^{\phi_n(t)}(\bar{\mathbf{P}}_n), \label{general_form}
\end{align}
where \eqref{general_form} is obtained by substituting \eqref{general_estimater} and \eqref{sys}  into \eqref{covariance1} and
\begin{equation}\label{eq:v}
\zeta_n(\mathbf{X})\triangleq \mathbf{A}_n\mathbf{X}\mathbf{A}_n^{\top}+\mathbf{W}_n
\end{equation}
$$\zeta_n^{1}(\cdot) \triangleq \zeta_n(\cdot), \quad \zeta_n^{m+1}(\cdot)  \triangleq \zeta_n (\zeta_n^{m}(\cdot)),\quad m\geq 1. $$ 
Thus, the quality of the remote estimation error of process $n$ in time slot $t$ can be quantified via $\text{Tr}\left(\mathbf{P}_n(t) \right)$. 
For ease of exposition, we introduce the following function
\begin{equation}\label{eq:c}
c_n(i)\triangleq \text{Tr}\left(\zeta^i_n(\bar{\mathbf{P}}_n)\right), \forall i\in \mathbb{N}
\end{equation}
and note that, using \eqref{general_form}, we can write
\begin{equation} \label{trace}
\text{Tr}\left(\mathbf{P}_n(t) \right) \triangleq c_n(\phi_n(t)).
\end{equation}
Therefore, the estimation quality of process $n$ at time $t$ is a function of its AoI state $\phi_n(t)$.

Propositions~1 and 2 of \cite{liu2020remote} allow us to state the following property of $c_n(\cdot)$:
\begin{lemma}\label{lem:c}
	For any $\epsilon>0$, there exists $N'>0$, $\kappa >0$ and $\eta>0$ such that $$c_n(i) < \kappa \left(\rho^2(\mathbf{A}_n)+\epsilon \right)^i, \forall i>N',$$ and 
	$$c_n(i) \geq \eta (\rho(\mathbf{A}_n))^{2i}, \forall i>N'. $$
\end{lemma}

In this work, we solely focus on deterministic stationary scheduling policies.
Let $\nu_n(t) \in \{0,1,\dots,M\}$, $n\in\mathcal{N}$, denote the selected frequency channel for process $n$ at time $t$. The sensor is not scheduled for transmission if $\nu_n(t)=0$.
Since available system states include the current AoI states and the previous channel states,  scheduling polices $\pi(\cdot)$ can be written as
\begin{equation}\label{eq:schedule_0}
\bm\nu(t) = \pi(\bm\phi(t),\mathbf{h}(t-1)),
\end{equation}
where $\bm \phi(t) \triangleq [\phi_1(t),\dots,\phi_N(t)]\in \mathbb{N}^N$ and $\bm\nu(t) = [\nu_1(t),\dots,\nu_N(t)]$.

{Note that  $\mathbf{P}_n(t) $ is a  countable stochastic process taking values from a countably infinite set $$\{\zeta^1_n(\bar{\mathbf{P}}_n),\zeta^2_n(\bar{\mathbf{P}}_n),\dots\}.$$
 If $\rho(\mathbf{A}_n)\geq 1$, then this process will grow during periods of consecutive packet dropouts. Since, due to fading, periods of consecutive packet dropouts have unbounded support, at best one can hope for some type of stochastic stability. Our focus is on  mean-square stability.}

\begin{definition}[Average Mean-Square Stability]
	The $N$-sensor remote estimator over $M$ frequency channels described above is average mean-square stabilizable, if there exists a deterministic stationary policy \eqref{eq:schedule_0}
	such that the sum average estimation mean-square error (MSE) $J$ is bounded, where
	\begin{equation}\label{longterm}
	J\triangleq \sum_{n=1}^{N} J_n
	\end{equation}
	and 
	\begin{equation}
	J_n = \limsup_{T\to\infty}\frac{1}{T}\sum_{t=1}^{T} \mathsf{E}\left[\text{\normalfont Tr}\left(\mathbf{P}_n(t)\right)\right],n\in\mathcal{N}.
	\end{equation}
	
\end{definition}

\section{Key Results} \label{sec:key}
In this section, we present and prove the necessary and sufficient condition for stabilizing the multi-sensor-multi-channel remote estimator in terms of the multi-process parameters and the multi-channel statistics.
As will become apparent establishing such a necessary and sufficient stability condition is highly non-trivial as we consider  transmission scheduling for multiple sensor packets over multiple fading channels that are correlated in both time and frequency domains.

\subsection{Stability Condition}
The physical process and channel parameters jointly determine the stability of the overall remote estimator. Our result is stated in terms of the $\bar{M}\times \bar{M}$  probability matrix $\mathbf{E}(\mathbf{v})$  obtained from  the 	channel state transition probability matrix $\mathbf{M}$: 
	\begin{equation} \label{eq:Ev}
\begin{aligned}
	\left[\mathbf{E}(\mathbf{v})\right]_{i,j}
	& \triangleq
	\myprob{\mathbf{h}(t) = \mathbf{h}_j, \mathbf{h}_j \notin \tilde{\mathcal{S}}_{v_i}  \vert \mathbf{h}(t-1) = \mathbf{h}_i}\\
	&=
	  p_{i,j} \mathbbm{1}(\mathbf{h}_j \notin \tilde{\mathcal{S}}_{v_i}),
\end{aligned}
	\end{equation}
\mycomment{where $\mathbf{v}$ is a length-$\bar{M}$ vector and $v_i \in\mathcal{M}$ is the $i$th element of $\mathbf{v}$ denoting the index of selected frequency channel with the observation of the vector channel state $\mathbf{h}_i$.
So $\left[\mathbf{E}(\mathbf{v})\right]_{i,j}$ denotes the probability that the current channel state is $\mathbf{h}_j$ and the packet transmission fails given the previously observed channel state $\mathbf{h}_i$ and the selected frequency channel for transmission $v_i$.}

\begin{theorem}\label{theo:main}
	Consider  $\rho_{\max} \triangleq \rho(\mathbf{A}_1)$, and \begin{equation}\label{eq:lambda_inf}
	\lambda_{\infty} \triangleq \liminf\limits_{L\rightarrow \infty} \lambda_L = \min_{L\in\mathbb{N}} \lambda_L,
	\end{equation}
	\begin{equation}\label{eq:lambda_main}
	\lambda_L \triangleq \min_{\mathbf{v}_l \in \mathcal{M}^{\bar{M}}} \rho\left(\mathbf{E}(\mathbf{v}_1) \mathbf{E}(\mathbf{v}_2) \cdots \mathbf{E}(\mathbf{v}_L)\right)^{\frac{1}{L}}
	\end{equation}
	where $\mathbf{v}_l \triangleq [v_{l,1},v_{l,2},\dots,v_{l,\bar{M}}]^\top \in \mathcal{M}^{\bar{M}}$ is a vector of frequency channel selection at the $\bar{M}$ different channel conditions.
	\par 
	The remote estimator described in Section~\ref{sec:sys} under Assumptions~\ref{ass:pre-channel} and~\ref{ass:no-redundent} is stable if and only if
	\begin{equation}\label{eq:key}
	\rho^2_{\max} \lambda_{\infty} <1.
	\end{equation}
\end{theorem}

Theorem~\ref{theo:main} shows that the stability depends on the spectral radius of the most unstable process and a complex function of the channel state transition probability matrix.	
Provided the condition is satisfied, one can always find a scheduling policy that stabilizes the remote estimator; if \eqref{eq:key} does not hold, then  no stabilizing scheduling policy exists.
Theorem~\ref{theo:main} does not provide direct insights  on the structure of a suitable scheduling policy. However, we will construct a policy with stability guarantees in the proof of the sufficiency part. 

We note that $\lambda_{\infty}$ can be treated as an  inverse quality indicator of the parallel correlated channels:  smaller values of $\lambda_{\infty}$  indicate a  better channel quality.
\mycomment{The infinity in $\lambda_{\infty}$ takes in to account the infinitely many different orders of matrix products in \eqref{eq:lambda_main}.}
\mycomment{The last equality in \eqref{eq:lambda_inf} can be easily obtained by using the property that $\rho(\mathbf{D}^L)^{\frac{1}{L}} =\rho(\mathbf{D})$ holds for any square matrix $\mathbf{D}$ and positive integer $L$.}

\begin{remark}[Computations]
\mycomment{Although $\lambda_{\infty}$ in \eqref{eq:lambda_inf} may not be achieved with a finite $L$, one can approximate $\lambda_{\infty}$ by finding the minimum value in $\{\lambda_1,\dots,\lambda_L\}$ in \eqref{eq:lambda_main}. 
Since $\lambda_\infty\leq \min\{\lambda_1,
\dots,\lambda_L\},\forall L\in\mathbb{N}$, we have the sufficient stability condition 
$$\rho^2_{\max} \min\{\lambda_1,
\dots,\lambda_L\} <1,$$ which approximates the necessary and sufficient condition \eqref{eq:key}  when $L$ is large.

In Fig.~\ref{fig:L}, we randomly generate six matrix sets each consisting of eight non-negative $4$-by-$4$ matrices; for each matrix set, we plot $\lambda_L$ as the minimum $L$th root of the spectral radius of the  $L$ matrix product, where each matrix is taken from the matrix set as in \eqref{eq:lambda_main}. It can be observed that $\lambda_L$ may and may not monotonically decrease with the increasing $L$, and a small $L$ (e.g., $L=2$) can almost achieve the minimum value obtained by a large $L$, i.e., $\lambda_2 \approx \min\{\lambda_1,\lambda_2,\dots,\lambda_5\}$ in different cases.
In some cases, $\lambda_1=\lambda_2=\dots=\lambda_5$.

}

\begin{figure}[t]
	\centering
	\vspace{-0.5cm}	
	\includegraphics[scale=0.6]{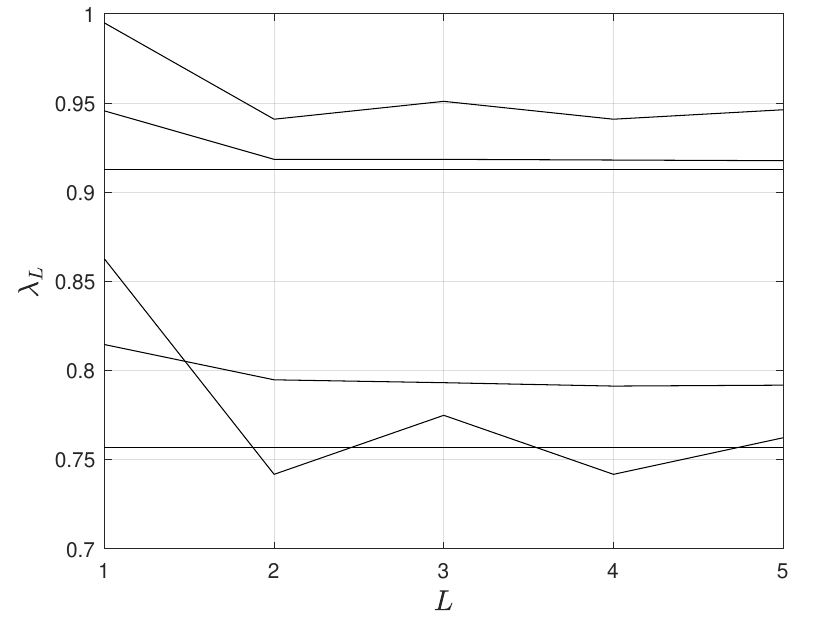}
	\caption{$\lambda_L$ versus $L$ over $6$ randomly generated matrix sets.}
	\vspace{-0.5cm}	
	\label{fig:L}
\end{figure}

\end{remark}

To the best of our knowledge, Theorem~\ref{theo:main} is the first necessary and sufficient stability condition for remote estimation over a multi-sensor-multi-channel network in the literature.	Before proving our result, we will first establish a relationship to existing results by focusing on a special case. Under the idealized assumption that the Markov channels at different frequencies are independent, a sufficient condition for stability was proved in \cite{LEONG2020108759}.
It corresponds to Corollary~\ref{coro:alex} below.

\begin{corollary}\label{coro:alex}
Consider the special case that Markov channels at different frequencies are independent,
and assume that the $\{0,1\}$ channel state transition probability matrix of frequency channel $m$ is 
\begin{equation}\label{eq:temp_M}
\mathbf{M}^{(m)} \triangleq 
\begin{bmatrix}
\alpha^{(m)}_{00} && \alpha^{(m)}_{01} \\
\alpha^{(m)}_{10} && \alpha^{(m)}_{11} \\
\end{bmatrix}, m \in \mathcal{M}.
\end{equation}
The remote estimator described in Section~\ref{sec:sys} under Assumptions~\ref{ass:pre-channel} and \ref{ass:no-redundent} is stable if~\cite{LEONG2020108759}
\begin{equation}\label{eq:cor1}
\rho_{\max}^2 \alpha^{(m^\star)}_{00} <1,
\end{equation}
where $\alpha^{(m^\star)}_{00} \triangleq \min_{m} \alpha^{(m)}_{00}$. 
\par The sufficient stability condition~\eqref{eq:cor1} is more restrictive than Theorem~\ref{theo:main}. 
\end{corollary}
\begin{proof}	
We only need to show that $\alpha^{(m^\star)}_{00} =\rho(\mathbf{E}(\mathbf{v}'_{m^\star}))$, where $\mathbf{v}'_{m^\star} \triangleq [m^\star,\dots,m^\star]^\top$, as $\rho(\mathbf{E}(\mathbf{v}'_{m^\star})) \geq \min_{L\in\mathbb{N}}  \min_{\mathbf{v}_l \in \mathcal{M}^{\bar{M}}} \rho\left(\mathbf{E}(\mathbf{v}_1) \mathbf{E}(\mathbf{v}_2) \cdots \mathbf{E}(\mathbf{v}_L)\right)^{\frac{1}{L}}$.	
Without loss of generality, we assume the first half of the channel states in $\mathcal{S}$ have the off state in the $m^\star$th channel, i.e.,
\begin{equation}
	\mathbf{h}_i \notin \mathcal{\tilde{S}}_{m^\star}, \forall i \in \{1,2,\dots,\bar{M}/2\}.
\end{equation}
Then, from  \eqref{eq:Ev}, $\mathbf{E}(\mathbf{v}'_{m^\star})$ can be written as a block lower triangular matrix
\begin{equation}
	\mathbf{E}(\mathbf{v}'_{m^\star})= \begin{bmatrix}
		\mathbf{E}_{00} & \aug& \mathbf{0}\\
		\hline
		\mathbf{E}_{10} & \aug& \mathbf{0}\\
	\end{bmatrix}
\end{equation}
and thus $\rho(\mathbf{E}(\mathbf{v}'_{m^\star})) = \rho(\mathbf{E}_{00})$. Further, it is easy to see that
\begin{equation}
	\sum_{j=1}^{\bar{M}/2} [\mathbf{E}_{00}]_{i,j} = \sum_{j=1}^{\bar{M}/2} \myprob{\mathbf{h}_j \vert \mathbf{h}_i} =\alpha^{(m^\star)}_{00}, i = 1,2,\dots,\bar{M}/2,
\end{equation}	
where the second equality is due to that the frequency channel $m^\star$ is independent to the other channels and the first $\bar{M}/2$ of the (vector) channel states contains all the possible channel states of the rest of the $M-1$ frequency channels.
Using the Perron-Frobenius Theorem~\cite{Perron}, we have  $\rho(\mathbf{E}_{00}) = \alpha^{(m^\star)}_{00}$, which completes the proof.
\end{proof}

\begin{example}
\normalfont
We numerically compare the sufficient stability condition in~\cite{LEONG2020108759} and the necessary and sufficient condition in Theorem~\ref{theo:main} for a two frequency channel scenario with $\rho^2_{\max} =2$.
We consider different channel state transition matrices \eqref{eq:temp_M} for the second frequency channel with parameters: (a) $\alpha^{(2)}_{00}=0.9, \alpha^{(2)}_{11}=0.9$, (b) $\alpha^{(2)}_{00}=0.9, \alpha^{(2)}_{11}=0.1$, (c) $\alpha^{(2)}_{00}=0.6, \alpha^{(2)}_{11}=0.9$, and (d) $\alpha^{(2)}_{00}=0.6, \alpha^{(2)}_{11}=0.1$.
A larger $\alpha^{(2)}_{00}$ and $\alpha^{(2)}_{11}$ lead to a longer channel state memory in the `off' and the `on' state, respectively.
Then the stability regions in terms of the parameters of first frequency channel state transition matrices, $\alpha^{(1)}_{00}$ and $\alpha^{(1)}_{11}$, are shown in Fig.~\ref{fig:compare_alex}.

It is clear that Theorem~\ref{theo:main} has a larger stability region than~\cite{LEONG2020108759} in case (c), which corresponds to that the second frequency channel has a shorter memory in the bad (`off') state and a longer memory in the good (`on') state. Hence when the quality of the second frequency channel is pretty good, the stability requirement on the first channel based on~\cite{LEONG2020108759} is more restrictive than Theorem~\ref{theo:main}. \hfill 
$\square$
\end{example}

\begin{figure}[t]
	\centering
	\includegraphics[scale=0.6]{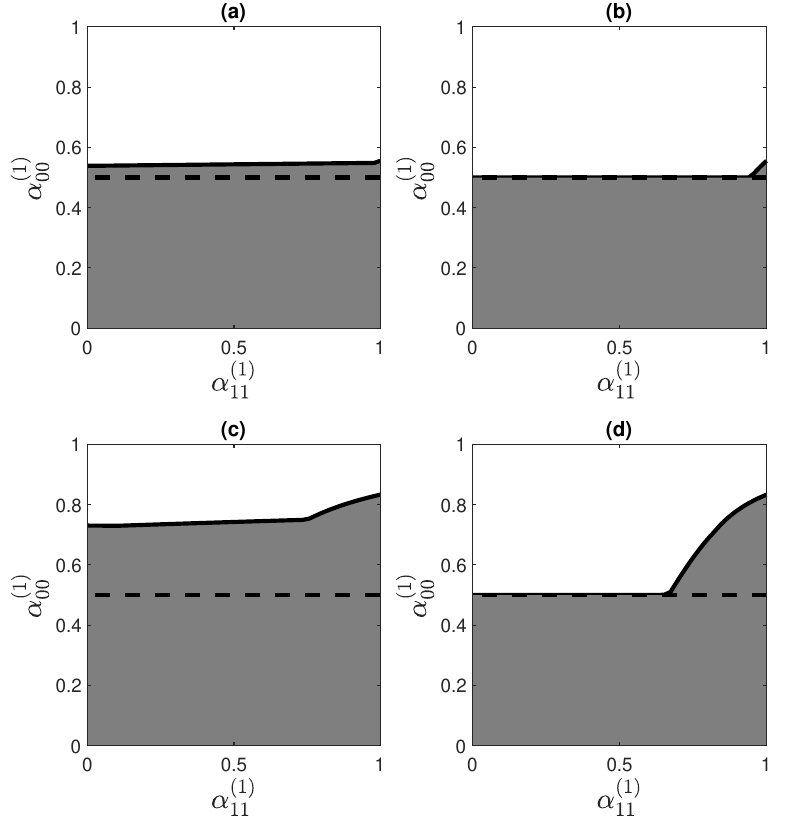}
	\caption{Comparison between stability conditions of Theorem~\ref{theo:main} (solid line) and~\cite{LEONG2020108759} (dashed line). The gray area denotes the parameter sets that satisfy the stability conditions.}
	\label{fig:compare_alex}
\end{figure}

We will prove the necessary and sufficiency parts of Theorem~\ref{theo:main} in the sequel.

\subsection{Proof of Necessity}
\subsubsection{Policy Construction}\label{sec:nec_policy}
To prove the necessity, we consider the scenario that only the estimate of the process with the largest spectral radius is scheduled for transmission in each time slot in a selected frequency channel, while the other sensors' estimates are perfectly known by the remote estimator and do not need any transmission. In other words, only one packet is scheduled in each time slot. We recall that  process~$1$ has the largest spectral radius, and we drop out the process index $n$ in the following analysis.

The channel scheduling policy~\eqref{eq:schedule_0} is reduced to
\begin{equation}\label{eq:schedule_1}
\nu(t) = \pi(\phi(t),\mathbf{h}(t-1)) \in \mathcal{M},
\end{equation}
where $\nu(t)$ and $\phi(t)$ denote $\nu_1(t)$ and $\phi_1(t)$, respectively, for ease of notation.

From \eqref{eq:schedule_1}, once the AoI $\phi$ is given, the channel selection rule given the previous channel state information can be written as
\begin{equation}
\mathbf{v}(\phi) = [v_1(\phi),\dots, v_{\bar{M}}(\phi)]\in \mathcal{M}^{\bar{M}},
\end{equation}
where $v_i(\phi) = \pi(\phi,\mathbf{h}_{i})$.

Given the channel selection vector $\mathbf{v}(\phi)$, we define the successful transmission probability matrix $\tilde{\mathbf{E}}(\mathbf{v}(\phi)) \in \mathbb{R}^{\bar{M}} \times \mathbb{R}^{\bar{M}}$, where
\begin{equation}
\begin{aligned}
[\tilde{\mathbf{E}}(\mathbf{v}(\phi))]_{i,j} 
&\triangleq \myprob{\mathbf{h}(t)=\mathbf{h}_j, \mathbf{h}_j \in \tilde{\mathcal{S}}_{v_i(\phi)} \vert \mathbf{h}(t-1)=\mathbf{h}_i} \\
&= p_{i,j} \mathbbm{1}(\mathbf{h}_j \in \tilde{\mathcal{S}}_{v_i(\phi)}).
\end{aligned}
\end{equation}
In other words, $[\tilde{\mathbf{E}}(\mathbf{v}(\phi))]_{i,j}$ is the probability that the current channel state is $\mathbf{h}_j$ and the transmission is successful in the selected $v_i(\phi)$th frequency channel given that the previous channel state is $\mathbf{h}_i$.
Accordingly, we define the failed transmission probability matrix 
\begin{equation}
{\mathbf{E}}(\mathbf{v}(\phi)) \triangleq \mathbf{M} -\tilde{\mathbf{E}}(\mathbf{v}(\phi)) \in \mathbb{R}^{\bar{M}} \times \mathbb{R}^{\bar{M}}.
\end{equation}

\subsubsection{Analysis of the Average Cost}\label{sec:nec_analysis}
Similar to \cite{liu2020remote}, we consider an estimation cycle based analysis  method.
Each estimation cycle starts after a successful transmission and ends at the next one, and thus the AoI state $\phi$ at the beginning of each estimation cycle is equal to $1$.
$T_k$ is the sum of transmissions in the $k$th estimation cycle.
$C_k$ is the sum MSE in the $k$th estimation cycle and is a function of $T_k$ as
\begin{equation}\label{g_fun}
C_k = g(T_k) \triangleq \sum_{j=1}^{T_k} c(j).
\end{equation}
The channel state before the $k$th cycle is denoted as $\mathbf{b}_k\in \mathcal{S}$, and a successful transmission occurs at $\mathbf{b}_k$.
Similar to Lemma~1 in \cite{liu2020remote}, we have the following Markovian property of the pre-cycle channel states.
\begin{lemma}\label{lem:G}
	$\{\mathbf{b}\}_{\mathbb{N}_0}$ is a time-homogeneous ergodic Markov chain with $\bar{M}_1\leq \bar{M}$ irreducible states of $\mathcal{S}$. 
	The state transition matrix of $\{\mathbf{b}\}_{\mathbb{N}_0}$ is $\mathbf{G}'$, which is the $\bar{M}_1$-by-$\bar{M}_1$ matrix taken from the top-left corner of
\begin{equation}
\mathbf{G} = \sum_{j=1}^{\infty} \mathbf{\tilde{\Xi}}(j),
\end{equation}
where
\begin{equation}\label{eq:Xi0}
\mathbf{\tilde{\Xi}}(j)= \mathbf{\mathbf{\Xi}}(j-1)\tilde{\mathbf{E}}(\mathbf{v}(j)), j=1,2,\dots.
\end{equation}
and
\begin{equation}\label{eq:Xi}
\mathbf{\mathbf{\Xi}}(j)= 
\begin{cases}
\mathbf{I},& j =0\\
\prod_{i=1}^{j} \mathbf{E}(\mathbf{v}(i)),& j >0.
\end{cases}
\end{equation}
The stationary distribution of $\{\mathbf{b}\}_{\mathbb{N}_0}$ is $\boldsymbol{\beta} \triangleq [\beta_1,\dots,\beta_{\bar{M}_1}]^\top$, which is the unique null-space vector of $(\mathbf{I-G'})^\top$ and $\beta_i >0,\forall i\in \mathcal{\bar{M}}_1$, where $\mathcal{\bar{M}}_1\triangleq \{1,2,\dots, \bar{M}_1\}$.
\end{lemma}
\begin{remark}
Only the first $\bar{M}_1$ channel states can be a pre-cycle state, and thus the last $(\bar{M}-\bar{M}_1)$ columns of $\mathbf{G}$ are all zeros. 
\end{remark}

From \eqref{longterm} it follows that  the average estimation MSE  can be rewritten as
\begin{align}\label{J2}
J =\limsup\limits_{K \rightarrow \infty} \frac{\frac{1}{K} (C_1+C_2+\dots+C_K)}{\frac{1}{K}(T_1+T_2+\dots+T_K)}
= \frac{\mathsf{E}\left[C\right]}{\mathsf{E}\left[T\right]},
\end{align}
where the last equality holds because   the distributions of $T_k$ and $C_k$ depend on $\mathbf{b}_k$, which is ergodic, and hence the unconditional distributions of $T_k$ and $C_k$ are also ergodic; time averages are equal to the ensemble averages and we drop the time indexes.
Then, we have
\begin{equation}\label{expect_T}
\myexpect{T}=\lim\limits_{K\rightarrow \infty }\frac{1}{K} \sum_{k=1}^{K} T_k =  \sum_{m=1}^{\bar{M}} \beta_m \myexpect{T \vert \mathbf{b}=\mathbf{h}_m},
\end{equation}
and
\begin{equation}\label{expect_C}
\myexpect{C} =\lim\limits_{K\rightarrow \infty }\frac{1}{K} \sum_{k=1}^{K} C_k =  \sum_{m=1}^{\bar{M}} \beta_m \myexpect{C \vert \mathbf{b}=\mathbf{h}_m},
\end{equation}
where $\beta_m$ is defined in Lemma~\ref{lem:G} when $m\in \mathcal{\bar{M}}_1$, and $\beta_m=0$ when $m \in  \mathcal{\bar{M}}_0 \triangleq   \mathcal{\bar{M}} \backslash\mathcal{\bar{M}}_1 = \{\bar{M}_1+1,\dots,\bar{M}\}$.

From the definition of estimation cycle and the property of channel state transition, the conditional probability of the length of an estimation cycle is obtained as
\begin{equation}\label{prob}
\begin{aligned}
&\myprob{T\!=\!i \vert \mathbf{b}\!=\!\mathbf{h}_m}
\!=\!  \vartheta\left(\mathbf{L}_m \mathbf{\tilde{\Xi}}(i)\right)
\!\triangleq\! \sum_{k=1}^{\bar{M}}\!\! \left[\mathbf{\tilde{\Xi}}(i)\right]_{m,k}\!, \forall m \in \mathcal{\bar{M}}
\end{aligned}
\end{equation}
where $\mathbf{L}_m$ is an all-zero matrix except for the $m$th diagonal element, which equals to~$1$.

Taking \eqref{prob} into \eqref{expect_T} and into \eqref{expect_C}, then   after some algebraic manipulations, one can obtain
\begin{align}\label{T}
\myexpect{T}&=  \sum_{m=1}^{\bar{M}} \beta_m \left(\sum_{i=1}^{\infty} i \vartheta\left(\mathbf{L}_m \mathbf{\tilde{\Xi}}(i)\right)\right),\\ \label{C}
\myexpect{C}&= \sum_{m=1}^{\bar{M}} \beta_m \left(\sum_{i=1}^{\infty} g(i) \vartheta\left(\mathbf{L}_m \mathbf{\tilde{\Xi}}(i)\right)\right).
\end{align}

From the definition of $g(i)$ in \eqref{g_fun} and the property of $c(i)$ in Lemma~\ref{lem:c}, $g(i)$ grows exponentially fast with $i$. It is easy to verify the property below.
\begin{lemma}\label{lem:C_only}
	$J < \infty$ if and only if $\myexpect{C} < \infty$.
\end{lemma}
From Lemma~\ref{lem:C_only}, it follows that it suffices to only investigate conditions such that $\myexpect{C}<\infty$ in the following.
%

\subsubsection{Proof of Necessity}
We define a set of channel selection vectors $\mathcal{F} \subset \mathcal{M}^{\bar{M}}$, where for any $\mathbf{v}\in\mathcal{F}$ we have $\mathbf{\tilde{E}}(\mathbf{v})=\mathbf{0}$ leading to zero chance of successful transmission in any of the frequency channels. Thus, for any $\mathbf{v}\in \mathcal{\tilde{F}} \triangleq \mathcal{M}^{\bar{M}}\backslash \mathcal{F}$, we have $\mathbf{\tilde{E}}(\mathbf{v}) \neq \mathbf{0}$.
It is clear that if $\tilde{\mathcal{F}}=\emptyset$, the packet dropout occurs all the time at all frequency channels. 
From the definition~\eqref{eq:lambda_main}, we can prove that $\lambda_{\infty}=1$ as $\mathbf{E}(\mathbf{v})$ is a stochastic matrix for any $\mathbf{v}\in\mathcal{M}^{\bar{M}}$.
Thus, the necessary condition of stability is straightforward as $\rho^2(\mathbf{A})<1$, which can be written as $\rho^2(\mathbf{A})\lambda_{\infty}<1$.
In what follows, we will focus on the scenario with $\tilde{\mathcal{F}}\neq \emptyset$.

We categorize all possible scheduling policies into two types:
\begin{definition}[Type-I and II Policies]
For a type-I policy, there exists $\bar{\phi}>0$ such that 
\begin{equation}
\mathbf{v}(\phi)\in \mathcal{F}, \forall \phi>\bar{\phi},
\end{equation}
For a type-II policy, if $\mathbf{v(\phi)} \in \mathcal{\tilde{F}}$, one can always find $\phi'>\phi$ such that $\mathbf{v(\phi')} \in \mathcal{\tilde{F}}$.
\end{definition}
A type-I policy has a strictly zero chance of successful transmission when the AoI is larger than a threshold, while a type-II policy still has a non-zero success probability when the AoI is arbitrarily large.
Thus, to stabilize the system, a type-I policy should guarantee successful transmission within the first $\bar{\phi}$ attempts, while it is not necessary for a type-II policy.

We first prove the necessary condition for type I and then type II policies.
Before proceeding further, we define the following
\begin{equation}\label{eq:lambda_m}
\lambda_{m,\infty} \triangleq \liminf\limits_{L\rightarrow \infty} \lambda_{m,L}, m \in \mathcal{\bar{M}},
\end{equation}
and
\begin{equation}
\lambda_{m,L} \triangleq \min_{\mathbf{v}_l \in \mathcal{M}^{\bar{M}}} \rho\left(\mathbf{L}_m \mathbf{E}(\mathbf{v}_1) \mathbf{E}(\mathbf{v}_2) \cdots \mathbf{E}(\mathbf{v}_L)\right)^{\frac{1}{L}}.
\end{equation}

\underline{Type-I Policy.} Since the scheduling policy consistently chooses channel selection vectors from $\mathcal{F}$ in the high AoI scenario that leads to zero chance of success,
it is clear that a necessary condition to stabilize the system is that the transmission process has a zero probability to fail for consecutive $\bar{\phi}$ times at the beginning of an estimation cycle. This can be written as
\begin{equation}
\max_{m \in \mathcal{\bar{M}}} \vartheta(\mathbf{L}_m \mathbf{E}(\mathbf{v}(1))\mathbf{E}(\mathbf{v}(2))\cdots\mathbf{E}(\mathbf{v}(\bar{\phi}))) = \mathbf{0},
\end{equation}
and hence $\max_{m \in \mathcal{\bar{M}}}  \lambda_{m,\infty} =0$. In the following, we focus on the type-II policy.

\underline{Type-II Policy.} We categorize the channel states $\mathbf{h}$ in two cases: (i) the pre-cycle states $\mathbf{h}=\mathbf{h}_m$ with $\beta_m>0$, i.e., $m\in \mathcal{\bar{M}}_1$, and 
(ii) the non-pre-cycle states $\mathbf{h}=\mathbf{h}_{m}$ with $\beta_{m}=0$, i.e., $m\in \mathcal{\bar{M}}_0$.
In other words, an estimation cycle can and cannot start after a case (i) and a case (ii) channel state, respectively.
Due to the ergodicity of the channel states, both the cases of channel states occur with non-zero probabilities.
Then, we analyze the necessary conditions to (i) make the average sum MSE  of an estimation cycle in \eqref{C} bounded that starts after a pre-cycle state
\begin{equation}\label{eq:C_new}
\sum_{j=1}^{\infty} g(j) \vartheta\left(\mathbf{L}_m \mathbf{\tilde{\Xi}}(j)\right)<\infty
\end{equation}
and
to (ii) make the average sum MSE of an estimation cycle bounded that contains a non-pre-cycle state.

(i) Assume that the channel state $\mathbf{h}_m$ is a pre-cycle state.  Using Perron–Frobenius Theorem~\cite{Perron}, we have 
\begin{equation}\label{eq:case_1}
\vartheta\left(\mathbf{L}_m \mathbf{\mathbf{\Xi}}(i)\right)\geq \rho\left(\mathbf{L}_m \mathbf{\mathbf{\Xi}}(i)\right),
\end{equation}
thus, there exists an element in $\mathbf{L}_m \mathbf{\mathbf{\Xi}}(i)$, e.g., the $m'$th element of the $m$th row, such that
\begin{equation} \label{eq:Lm}
\left[\mathbf{L}_m \mathbf{\mathbf{\Xi}}(i)\right]_{m,m'}\geq \frac{1}{\bar{M}}\rho\left(\mathbf{L}_m \mathbf{\mathbf{\Xi}}(i)\right).
\end{equation}

From \eqref{eq:lambda_m}, there is a constant $\kappa_0>0$ such that $\rho\left(\mathbf{L}_m \mathbf{\mathbf{\Xi}}(i)\right)\geq \kappa_0 (\lambda_{m,\infty})^{i}$.
Let $\mathbf{F}(i)\in\mathbb{R}^{\bar{M}\times \bar{M}}$ denote a degraded matrix of $\mathbf{L}_m \mathbf{\mathbf{\Xi}}(i)$, where $[\mathbf{F}(i)]_{m,m'} = [\mathbf{L}_m \mathbf{\mathbf{\Xi}}(i)]_{m,m'}$ and the other elements of $\mathbf{F}(i)$ are zeros.
From \eqref{eq:Lm}, it is clear that
\begin{equation}\label{eq:vF}
\vartheta\left(\mathbf{L}_m \mathbf{\mathbf{\Xi}}(i)\right)\geq \vartheta\left(\mathbf{F}(i)\right) \geq \frac{\kappa_0}{\bar{M}}\left(\lambda_{m,\infty}\right)^{i}.
\end{equation}

Based on the property of type-II policy, we define an infinite sequence of AoI variables $\{\tilde{\phi}_1,\tilde{\phi}_2,\dots\}$, where $\tilde{\phi}_i < \tilde{\phi}_j$ if $i<j$, and $\mathbf{v}(\phi)\in \tilde{\mathcal{F}}$ if and only if $\phi\in \{\tilde{\phi}\}_\mathbb{N}$. 
Also, we define the operator $\lceil x \rceil_{\mathcal{A}}$ as the smallest value in the set $\mathcal{A}$ that is no smaller than $x$.
Building on $\{\tilde{\phi}_1,\tilde{\phi}_2,\dots\}$ and a constant $L \in \mathbb{N}$, we construct a sequence of AoI as $\{\varphi_1,\varphi_2,\dots\}$, where
\begin{equation}\label{eq:varphi}
\varphi_i = \begin{cases}
\lceil L \rceil_{\{\tilde{\phi}\}_\mathbb{N}}, & i=1\\
\lceil \varphi_{i-1}+L \rceil_{\{\tilde{\phi}\}_\mathbb{N}}, & i>1,
\end{cases}
\end{equation}
and thus $\mathbf{v}(\varphi_i)\in \mathcal{\tilde{F}},\  \varphi_{i+1}-\varphi_i\geq L, \forall i\in\mathbb{N}$.

Then, we introduce the following technical lemma.
\begin{lemma} \label{lem:nec_pos}
	There is a constant $L>0$ such that 
	the transmission is successful with at least a non-zero probability $\xi$ within the next $L$ steps for any channel selection vector at the $L$th step $ \mathbf{v}_L \in \mathcal{\tilde{F}}$, no matter what the current channel state is and the first $(L-1)$-step channel selection vectors are, i.e., 
	\begin{equation}
	\begin{aligned}
	\xi &\triangleq \min_{\mathbf{v}_1,\mathbf{v}_2,\dots,\mathbf{v}_{L-1} \in \mathcal{M}^{\bar{M}}, \mathbf{v}_L\in\mathcal{\tilde{F}}}
	\min_{m \in \mathcal{\bar{M}}}
	\max_{l=1,\dots,L}\\ &\qquad\qquad\qquad\quad\vartheta(\mathbf{L}_m\mathbf{E}(\mathbf{v}_1)\mathbf{E}(\mathbf{v}_2)\cdots\mathbf{E}(\mathbf{v}_{l-1})\mathbf{\tilde{E}}(\mathbf{v}_l)) >0,
	\end{aligned}
	\end{equation}
	where $\mathbf{E}(\mathbf{v}_0)\triangleq \mathbf{I}$.
\end{lemma}
\begin{proof}
See Appendix.
\end{proof}

From Lemma~\ref{lem:nec_pos}, we can find a constant $L$ to construct the AoI sequence $\{\varphi\}_\mathbb{N}$ in~\eqref{eq:varphi}.
Then, using Lemma~\ref{lem:c} and \eqref{eq:vF}, we have
\begin{equation}
\begin{aligned}
&\sum_{j=\varphi_i-L}^{\varphi_i} g(j) \vartheta\left(\mathbf{L}_m \mathbf{\tilde{\Xi}}(j)\right)
\geq g(\varphi_i) \vartheta\left(\mathbf{L}_m \mathbf{\tilde{\Xi}}(\varphi_i)\right)\\
&\geq g(\varphi_i-L) \vartheta\Big(\mathbf{L}_m \mathbf{{\Xi}}(\varphi_i-L) \mathbf{E}(\mathbf{v}(\varphi_i-L+1)) \\
& \mathbf{E}(\mathbf{v}(\varphi_i-L+2))
\cdots
\mathbf{E}(\mathbf{v}(\varphi_i-1))
\mathbf{\tilde{E}}(\mathbf{v}(\varphi_i))\Big) \\
&\geq g(\varphi_i-L)\vartheta\left(\mathbf{L}_m \mathbf{{\Xi}}(\varphi_i-L) \right) \xi\\
&\geq \frac{\kappa_0 \xi}{\bar{M}}
(\rho(\mathbf{A})^2)^{\varphi_i-L}
\left(\lambda_{m,\infty}\right)^{\varphi_i-L}.
\end{aligned}
\end{equation}
Now the average sum MSE per cycle in \eqref{eq:C_new} is
\begin{equation} \label{eq:EC}
\begin{aligned}
&\sum_{j=1}^{\infty} g(j) \vartheta\left(\mathbf{L}_m \mathbf{\tilde{\Xi}}(j)\right)
\geq \sum_{i=1}^{\infty} \sum_{j=\varphi_i-L}^{\varphi_i} g(j) \vartheta\left(\mathbf{L}_m \mathbf{\tilde{\Xi}}(j)\right)\\
&\geq \sum_{i=1}^{\infty} \frac{\kappa_0\xi}{\bar{M}}
(\rho(\mathbf{A})^2)^{\varphi_i-L}
\left(\lambda_{m,\infty}\right)^{\varphi_i-L}.
\end{aligned}
\end{equation}

To make the last sum in \eqref{eq:EC} bounded, we must have $\lim\limits_{i \rightarrow \infty} (\rho(\mathbf{A})^2)^{\varphi_i}
\left(\lambda_{m,\infty}\right)^{\varphi_i} =0$, i.e., $(\rho(\mathbf{A})^2)\lambda_{m,\infty}<1$.
Thus, by considering all the pre-cycle channel states,
$$\max_{m \in \mathcal{\bar{M}}_1}\rho(\mathbf{A})^2 \lambda_{m,\infty}<1$$ holds if $\myexpect{C}$ in \eqref{C} is bounded.

(ii) Assume now that the channel state $\mathbf{h}_{m'}$ is a non-pre-cycle state. Since the channel state transition process is an ergodic Markov chain, given a state $\mathbf{h}_{m}$ with $m\in \mathcal{\bar{M}}_1$, it will take finite steps to arrive at $\mathbf{h}_{m'},m'\in \mathcal{\bar{M}}_0$ with a positive probability no matter what the channel scheduling policy is. Mathematically, there exits $L'\geq 1$ such that
$$ \tilde{\beta}_{m,m'} \triangleq [\mathbf{L}_{m} \mathbf{\mathbf{\Xi}}(L')]_{m,m'}>0.$$
Then, it is straightforward to have
\begin{equation}\label{eq:case_2}
\vartheta\!\left(\mathbf{L}_{m} \mathbf{\mathbf{\Xi}}(i)\right)
\!\!\geq \!
\tilde{\beta}_{m,m'} \vartheta\!\left(\!\mathbf{L}_{m'} \mathbf{E}(\mathbf{v}_{L'+1})\mathbf{E}(\mathbf{v}_{L'+2})\cdots \mathbf{E}(\mathbf{v}_{i})\!\right)\!,\! \forall i\!\geq\! L'\!.
\end{equation}
We see that the right-hand side of \eqref{eq:case_2} has a similar format to the left-hand side of \eqref{eq:case_1} in case~(i).
Then
by taking \eqref{eq:case_2} into \eqref{eq:C_new}, following   similar steps as in case (i) and considering all the non-pre-cycle channel states, we obtain that 
$$\max_{m' \in \mathcal{\bar{M}}_0}\rho(\mathbf{A})^2 \lambda_{m',\infty}<1$$ holds if $\myexpect{C}$ in \eqref{C} is bounded.

From cases (i) and (ii), a necessary condition to bound $\myexpect{C}$ can be  uniformly written as 
$$\max_{m \in \mathcal{\bar{M}}}\rho(\mathbf{A})^2 \lambda_{m,\infty}<1.$$
From the technical lemma below, it is easy to prove that $\max_{m \in \mathcal{\bar{M}}} \lambda_{m,\infty} = \lambda_{\infty}$. 
Then, the necessary condition of type-I and II policies can be jointly written as $\rho(\mathbf{A})^2 \lambda_{\infty}<1$, which completes the proof of the necessity of Theorem~\ref{theo:main}.

\begin{lemma}\label{lem:tech}
	Given a sequence of $N$-by-$N$ matrices $\{\mathbf{M}_1,\mathbf{M}_2,\dots\}$, the following equation holds
	\begin{equation}
	\max_{i\in \mathcal{N}} \liminf\limits_{L \rightarrow \infty}\!\rho\left(\mathbf{L}_i\mathbf{M}_1\mathbf{M}_2\cdots\mathbf{M}_L\right)^{\frac{1}{L}}
	\!=\!
	\liminf\limits_{L \rightarrow \infty}\!\rho\left(\mathbf{M}_1\mathbf{M}_2\cdots\mathbf{M}_L\right)^{\frac{1}{L}}\!,
	\end{equation}
	where $\mathbf{L}_i\in\mathbb{R}^{N\times N}$ is a diagonal matrix with $i$th diagonal element equals to $1$ and the other elements are zeros.
\end{lemma}
\begin{proof}
See Appendix.
\end{proof}

\subsection{Proof of Sufficiency}
\subsubsection{Policy Construction}\label{sec:suf_policy}
We consider a \emph{persistent serial scheduling} policy that persistently schedules the transmission of one sensor at a time until it is successful and then schedules the next sensor and so on. Although it seems that such a policy cannot take advantage of the parallel channels, we will show that the policy stabilizes the remote estimator if condition \eqref{eq:key}  holds.
The policy can be written as
\begin{equation}
\pi(\bm{\phi}(t),\mathbf{h}(t-1)) \!= \!\!
\begin{cases}
\pi_1(\phi_1(t),\mathbf{h}(t-1)), & \text{if } \phi_n(t)\!=\!1, n\!=\!N\\
\pi_{n+1}(\phi_{n+1}(t),\mathbf{h}(t-1)), & \text{if } \phi_n(t)\!=\!1, n\!<\!N\\
\pi(\bm{\phi}(t-1),\mathbf{h}(t-2)),& \text{else } \\
\end{cases}
\end{equation}
where the policy
\begin{equation}\label{eq:policy_long}
\pi_n(\phi_n(t),\mathbf{h}(t-1))=[\underbrace{0,\dots,0}_{n-1},\nu_n(t),\underbrace{0,\dots,0}_{N-n}], n\in\mathcal{N}
\end{equation}
denotes scheduling the $n$th sensor on the $\nu_n(t)$th frequency channel.
Without loss of generality,
the initial scheduling policy is given by
$\pi(\bm{\phi}(1),\mathbf{h}(0))=\pi_1({\phi}_1(1),\mathbf{h}(0))$.
With a slight abuse of notation, we drop out the zeros in \eqref{eq:policy_long} so that the frequency channel selection rule for sensor $n$ is rewritten as
\begin{equation}\label{eq:policy_short}
\pi_n(\phi_n(t),\mathbf{h}(t-1))=\nu_n(t)\in \mathcal{M}, n\in\mathcal{N}.
\end{equation}
Once the AoI of sensor $n$, $\phi_n$, is given, the channel selection rule with different previous channel states is denoted as
\begin{equation}
\mathbf{v}_n(\phi_n)\triangleq\left[\pi_n(\phi_n,\mathbf{h}_1),\pi_n(\phi_n,\mathbf{h}_2),\dots,\pi_n(\phi_n,\mathbf{h}_{\bar{M}}) \right] \in \mathcal{M}^{\bar{M}}.
\end{equation}


Moreover, we assume that the frequency channel selection rule $\mathbf{v}_n(\phi_n),\forall n \in \mathcal{N}$, is a \emph{periodic policy} in terms of $\phi_n$ with $L$ potential channel selection vectors taken from the set  $\mathcal{V}\triangleq\{\mathbf{\tilde{v}}_n(1),\mathbf{\tilde{v}}_n(2),\dots, \mathbf{\tilde{v}}_n(L)\}$, where $\mathbf{\tilde{v}}_n(L)\in \mathcal{M}^{\bar{M}}$.
From the definition of $\lambda_{\infty}$ in \eqref{eq:lambda_inf},
for an arbitrarily small $\epsilon>0$, we can find a constant $L$ and a length-$L$ channel-selection-vector set $\mathcal{V}$ satisfying the condition that
\begin{equation}\label{eq:suf_rho}
 \rho\left(\mathbf{E}(\mathbf{\tilde{v}}_n(1))\mathbf{E}(\mathbf{\tilde{v}}_n(2))\cdots\mathbf{E}(\mathbf{\tilde{v}}_n(L)) \right)^{\frac{1}{L}}\leq \lambda_{\infty} + \epsilon.
\end{equation}
Then, the periodic channel selection policy of sensor $n$ is defined as
\begin{equation}
\mathbf{v}_n(\phi_n)= \mathbf{\tilde{v}}_n(\phi_n \bmod L),
\end{equation}
where $(a \bmod b)$ denotes the remainder of the Euclidean division of $a$ by $b$ if the remainder is non-zero, otherwise $(a \bmod b) =b$.

\subsubsection{Analysis of the Average Cost}
To analyze the average cost of the average estimation MSE of the $n$th process $J_n$, the estimation cycle starts after a successful transmission of process $n$ and ends at the following one.
$T_{n,k}$ and $C_{n,k}$ denote  the sum of transmissions and the sum MSE of the $k$th estimation cycle:
\begin{equation}\label{g_fun2}
C_{n,k} = g(T_{n,k}) \triangleq \sum_{j=1}^{T_{n,k}} c_n(j),
\end{equation}
and 
\begin{equation}
T_{n,k} = \sum_{i=1}^{N} T_{(n,i),k},
\end{equation}
where $T_{(n,i),k}$ is the time duration scheduled for the sensor $i$'s transmission during the $k$th estimation cycle of process $n$.

Similar to the single-sensor case, it can be proved that the time average of $\{C_n\}_{\mathbb{N}_0}$ is equal to its ensemble average, and $J_n$ is bounded if $\myexpect{C_n}$ is. In the following, we drop out the time index $k$ and will analyze $\myexpect{C_n}$.
For the ease of notation but without loss of generality, we analyze $\myexpect{C_N}$ for  the $N$th process:
\begin{equation}\label{eq:sum}
\begin{aligned}
\myexpect{C_N} &\!=\!\! \sum_{t_1=1}^{\infty}\dots\!\!\!\! \sum_{t_{N-1}=1}^{\infty}\!\!
\!\myexpect{C_N \vert T_{(N,1)}\!=\!t_1,\dots,T_{(N,N-1)}\!=\!t_{N-1}} \\
&\times \myprob{T_{(N,1)}=t_1,\dots,T_{(N,N-1)}=t_{N-1}}.
\end{aligned}
\end{equation}

Let $\mathbf{b}_N$ denote the channel state before an estimation cycle of process $N$. From~\eqref{prob}, it is easy to obtain the conditional probability
\begin{equation}\label{eq:prob_con}
\begin{aligned}
&\myprob{T_{(N,N)}=t_N \vert T_{(N,1)}=t_1,\dots,T_{(N,N-1)}=t_{n-1}}\\
& = \vartheta(\text{diag}\{[\varsigma_1,\dots, \varsigma_{\bar{M}}]\}\mathbf{\tilde{\Xi}}(t_N)),
\end{aligned}
\end{equation}
where 
\begin{equation}
\varsigma_m \!\!\triangleq\! \myprob{\mathbf{b}_N\!=\!\mathbf{h}_m \vert T_{(N,1)}\!=\!t_1,\dots,T_{(N,N-1)}\!=\!t_{n-1}}\!\!, m\!\in\!\! \mathcal{\bar{M}}.
\end{equation}
Recall that $\mathbf{\tilde{\Xi}}(\cdot)$ was defined in~\eqref{eq:Xi0}.

Thus, from \eqref{g_fun2} and \eqref{eq:prob_con}, the conditional expectation is simplified as
\begin{equation}\label{eq:C_condition}
\begin{aligned}
&\myexpect{C_N \vert T_{(N,1)}=t_1,\dots,T_{(N,N-1)}=t_{N-1}}\\
&= \sum_{t_{N}=1}^{\infty} g\left(\sum_{i=1}^{N} t_{i}\right)\\
& \times \myprob{T_{(N,N)}=t_N \vert T_{(N,1)}=t_1,\dots,T_{(N,N-1)}=t_{n-1}}\\
&=\sum_{t_N=1}^{\infty} g\left(\sum_{i=1}^{N} t_i\right)  \vartheta(\text{diag}\{[\varsigma_1,\dots, \varsigma_{\bar{M}}]\}\mathbf{\tilde{\Xi}}(t_N)).
\end{aligned}
\end{equation}

\subsubsection{Proof of Sufficiency}
From Lemma~\ref{lem:c} and the monotonicity of $g(\cdot)$, it can be proved that for any $\epsilon'>0$ we can find $\kappa>0$ such that
\begin{equation}\label{eq:ineq_C}
	g\left(\sum_{i=1}^{N} t_i\right) \leq \kappa (\rho^2(\mathbf{A}_N)+\epsilon')^{\sum_{i=1}^{N} t_i}.
\end{equation}
Since $\varsigma_m \leq 1, \forall m\in\mathcal{\bar{M}}$, it is clear that
\begin{equation}\label{eq:ineq_v}
	\begin{aligned}
		& \vartheta(\text{diag}\{[\varsigma_1,\dots, \varsigma_{\bar{M}}]\}\mathbf{\tilde{\Xi}}(t_N))
		\leq \vartheta(\mathbf{\tilde{\Xi}}(t_N))\\
		&\leq \vartheta(\mathbf{\Xi}(t_N - (t_N\bmod L))), \text{ if } t_N>L,
	\end{aligned}
\end{equation}
where the last inequality is due to the property that
$\vartheta(\mathbf{M}_1\mathbf{M}_2) \leq \vartheta(\mathbf{M}_1)$ if $\mathbf{M}_1$ and $\mathbf{M}_2$ are non-negative matrices and the sum of each row of $\mathbf{M}_2$ is no higher than $1$.
%

Taking \eqref{eq:ineq_C} and \eqref{eq:ineq_v} into \eqref{eq:C_condition}, 
\begin{equation}
	\begin{aligned}
		&\myexpect{C_N \vert T_{(N,1)}=t_1,\dots,T_{(N,N-1)}=t_{n-1}} \\
		&\leq \gamma_0 + \kappa L \sum_{j=2}^{\infty}   (\rho^2(\mathbf{A}_N)+\epsilon')^{\sum_{i=1}^{N-1} t_{i}}\\
		&\hspace{3.5cm}\times(\rho^2(\mathbf{A}_N)+\epsilon')^{jL} \vartheta(\mathbf{\Xi}((j-1)L))\\
		&\leq \gamma_0 + \kappa' \kappa L (\rho^2(\mathbf{A}_N)+\epsilon')^{\sum_{i=1}^{N-1} t_{i}}\\ &\hspace{3cm} \times \sum_{j=2}^{\infty}   
		(\rho^2(\mathbf{A}_N)+\epsilon')^{jL} (\lambda_{\infty}+\epsilon)^{(j-1)L},
	\end{aligned}
\end{equation}
where $\gamma_0, \kappa$, and $\kappa'$ are positive constant; the first inequality is obtained by dividing the infinity sum into length-$L$ segment sums and applying the inequalities \eqref{eq:ineq_C} and \eqref{eq:ineq_v}. The second inequality is obtained by
first rewriting $\mathbf{\Xi}((j-1)L)$ as $\mathbf{E}(\mathbf{\tilde{v}}_n(1))\mathbf{E}(\mathbf{\tilde{v}}_n(2))\cdots\mathbf{E}(\mathbf{\tilde{v}}_n(L))$ to the power of $(j-1)L$ and then using \eqref{eq:suf_rho} and Lemma~\ref{lem:rho} below.
\begin{lemma}[\cite{liu2020remote}]\label{lem:rho}
	Given an $N$-by-$N$ matrix $\mathbf{M}$, for any $\epsilon >0$, we can find $\kappa>0$ such that $
	[\mathbf{M}^k]_{i,j} < \kappa (\rho(\mathbf{M})+\epsilon)^k, \forall k\in\mathbb{N},i,j=1,2,\dots,N. \hfill \square$ 
\end{lemma}

As a consequence, if $\rho^2(\mathbf{A}_N) \lambda_{\infty} <1$, then we can always find a scheduling policy providing an arbitrary small $\epsilon$ to make $\myexpect{C_N \vert T_{(N,1)}=t_1,\dots,T_{(N,N-1)}=t_{n-1}} $ bounded. Since the condition does not rely on $\{T_{(N,1)}, T_{(N,2)} \dots T_{(N,N-1)}\}$, it is straightforward to show that the policy leads to a bounded $\myexpect{C_N}$.

By applying the above method to the other $N-1$ sensors, a sufficient condition for stabilizing the $N$-sensor remote estimator is obtained as $\rho^2(\mathbf{A}_1) \lambda_{\infty} <1$, which completes the proof of sufficiency of Theorem~\ref{theo:main}.

\begin{remark}
The policy above with persistent sensor scheduling and periodic channel selection is a stability-guaranteeing policy.
Note that such a policy, which does not utilize the parallel frequency channels, is only constructed for the proof of the sufficiency of Theorem~\ref{theo:main}, and is not optimal.
\mycomment{Once the stability condition is satisfied, we can find the optimal policies in~\eqref{eq:schedule_0}, e.g., by designing suitable MDP problems, see for example~\cite{LEONG2020108759}.}
\end{remark}

\section{Conclusions}\label{sec:con}
We have tackled the open problem: what is the fundamental requirement on the multi-sensor-multi-channel system to guarantee the existence of a sensor scheduling policy that can stabilize   remote estimation?
To solve the problem, we have proposed novel policy construction methods, and have developed an estimation-cycle based analytical approach.
We have derived a necessary and sufficient stability condition in terms of the LTI system parameters and the channel statistics.
Numerical results have shown that the condition is more effective than existing sufficient conditions available  in the literature.
Scheduling policies with stability guarantees have been derived as well.
For future work, we will consider stability analysis of a multi-control-loop system over shared wireless channels.

\section*{Appendix}

\subsection{Proof of Lemma~\ref{lem:nec_pos}}
Due to the periodicity of the Markov channel states, given any current channel state, there exists a constant $(L-1)$ such that it can reach any channel state in the $(L-1)$th time slots with non-zero probabilities~\cite{durrett2019probability}, i.e.,
\begin{equation}\label{eq:M}
	[\mathbf{L}_m \mathbf{M}^{L-1}]_{m,m'}>0, \forall m,m'\in \mathcal{\bar{M}}.
\end{equation}

Given a current channel state $\mathbf{h}_m$ and a sequence of channel selection vectors $\{\mathbf{v}_1,\mathbf{v}_2,\dots,\mathbf{v}_{L}\}$ for the next $L$ transmissions,
if there is a non-zero probability $\xi_m$ that the transmissions can be successful within $(L-1)$ time slots, we directly have
\begin{equation}\label{eq:LE}
	\max_{l=1,\dots,L-1} \vartheta(\mathbf{L}_m\mathbf{E}(\mathbf{v}_1)\mathbf{E}(\mathbf{v}_2)\cdots\mathbf{\tilde{E}}(\mathbf{v}_l))>0.
\end{equation}
Otherwise, we have $$\mathbf{L}_m\mathbf{E}(\mathbf{v}_1)\mathbf{E}(\mathbf{v}_2)\cdots\mathbf{\tilde{E}}(\mathbf{v}_{l})= \mathbf{0}, \forall l=1,\dots,L-1.$$
Since $\mathbf{M} = \mathbf{E}(\mathbf{v}_l)+\mathbf{\tilde{E}}(\mathbf{v}_l), \forall l=1,\dots,L-1$,
it is clear that
\begin{equation}
	\mathbf{L}_m \mathbf{M}^{L-1} = \mathbf{L}_m\mathbf{E}(\mathbf{v}_1)\mathbf{E}(\mathbf{v}_2)\cdots\mathbf{E}(\mathbf{v}_{L}).
\end{equation}
By using \eqref{eq:M} and the fact that $\vartheta(\mathbf{\tilde{E}}(\mathbf{v}_L))\neq \mathbf{0}$, we have $\vartheta(\mathbf{L}_m\mathbf{E}(\mathbf{v}_1)\mathbf{E}(\mathbf{v}_2)\cdots\mathbf{E}(\mathbf{v}_{L-1})\mathbf{\tilde{E}}(\mathbf{v}_{L}))>0$. This completes the proof of Lemma~\ref{lem:nec_pos}.

\subsection{Proof of Lemma~\ref{lem:tech}}
We define the optimal index $i^\star$ such that
\begin{equation}
	\begin{aligned}
		i^\star \triangleq \arg\max_{i\in \mathcal{N}} \rho\left(\mathbf{L}_i\mathbf{M}_1\mathbf{M}_2\cdots\mathbf{M}_L\right).
	\end{aligned}
\end{equation}
By using Gelfand corollaries~\cite{horn2012matrix}, we have
\begin{equation}
	\rho\left(\mathbf{L}_{i^\star}\mathbf{M}_1\mathbf{M}_2\cdots\mathbf{M}_L\right)\leq \rho(\mathbf{L}_{i^\star})\rho(\mathbf{M}_{1} \mathbf{M}_{2}\cdots\mathbf{M}_{L})
\end{equation}
and thus
\begin{equation} \label{ieq:low}
	\begin{aligned}
		&\liminf\limits_{L \rightarrow \infty} \!\rho\left(\mathbf{L}_{i^\star}\mathbf{M}_1\mathbf{M}_2\cdots\mathbf{M}_L\right)^{\frac{1}{L}}\\
		&\!\!\leq \!
		\liminf\limits_{L \rightarrow \infty}
		\Big(\!\rho(\mathbf{L}_{i^\star})\rho(\mathbf{M}_{1} \mathbf{M}_{2}\cdots\mathbf{M}_{L})\!\Big)^{\!\frac{1}{L}} 
		\!\!\!=\! \liminf\limits_{L \rightarrow \infty}
		\!\rho(\mathbf{M}_{1}\cdots\mathbf{M}_{L})^{\!\frac{1}{L}}.
	\end{aligned}
\end{equation}

From the definition of $\mathbf{L}_i$, it is clear that
\begin{equation}
	\sum_{i=1}^{N} \mathbf{L}_i \mathbf{M}_1\mathbf{M}_2\cdots\mathbf{M}_L =\mathbf{M}_1\mathbf{M}_2\cdots\mathbf{M}_L.
\end{equation}
From the matrix norm property that the spectral radius of a sum of matrices is no larger than the sum of the spectral radius of the individuals~\cite{horn2012matrix}, we have
\begin{equation}
	\begin{aligned}
		\rho\left(\mathbf{M}_1\mathbf{M}_2\cdots\mathbf{M}_L\right) 
		\leq \sum_{i=1}^{N} \rho\left(\mathbf{L}_i \mathbf{M}_1\mathbf{M}_2\cdots\mathbf{M}_L\right)\\
		\leq N \rho\left(\mathbf{L}_{i^\star} \mathbf{M}_1\mathbf{M}_2\cdots\mathbf{M}_L\right),
	\end{aligned}
\end{equation}
and thus
\begin{equation}\label{ieq:up}
	\liminf\limits_{L \rightarrow \infty} \rho\left(\mathbf{M}_1\mathbf{M}_2\cdots\mathbf{M}_L\right)^{\frac{1}{L}}
	\leq \liminf\limits_{L \rightarrow \infty} \rho\left(\mathbf{L}_{i^\star} \mathbf{M}_1\mathbf{M}_2\cdots\mathbf{M}_L\right)^\frac{1}{L}.
\end{equation}

Using the inequalities \eqref{ieq:low} and \eqref{ieq:up}, this completes the proof of Lemma~\ref{lem:tech}.

    \balance
    
	\ifCLASSOPTIONcaptionsoff
	\newpage
	\fi

	\bibliographystyle{IEEEtran}

\end{document}